\documentclass[sigconf,nonacm]{acmart}
\AtBeginDocument{%
  }

\usepackage[colorinlistoftodos,prependcaption]{todonotes}

\usepackage{graphicx}
\usepackage{svg}
\usepackage{tikz}
\usepackage{subcaption}
\usetikzlibrary{arrows.meta, fit, backgrounds, decorations.pathreplacing, calc}
\usepackage{hyperref}
\usepackage[nameinlink,noabbrev,capitalize]{cleveref}

\usepackage{algorithm}
\usepackage{algpseudocode}
\algnewcommand{\Input}{\item[\textbf{Input:}]}
\algnewcommand{\Output}{\item[\textbf{Output:}]}
\newcommand{\algseparator}{\vspace{0.2em}\hrule\vspace{0.2em}}
\newcommand{\codeLink}{\url{https://github.com/RickHeijden/rangereach}}

\newtheorem*{remark}{Remark}

\newcommand{\bigO}{
    \ensuremath{\mathcal{O}}
}
\newcommand{\reach}{\rightsquigarrow}
\newcommand{\query}[1]{\textsc{#1}}

\begin{document}

\title{Even Faster Geosocial Reachability Queries}

\author{Rick van der Heijden}
\affiliation{%
  \institution{Eindhoven University of Technology}
  \city{Eindhoven}
  \country{The Netherlands}}
\email{r.h.j.v.d.heijden@student.tue.nl}

\author{Nikolay Yakovets}
\affiliation{%
  \institution{Eindhoven University of Technology}
  \city{Eindhoven}
  \country{The Netherlands}}
\email{hush@tue.nl}

\author{Thekla Hamm}
\affiliation{%
  \institution{Eindhoven University of Technology}
  \city{Eindhoven}
  \country{The Netherlands}}
\email{t.l.s.hamm@tue.nl}

\renewcommand{\shortauthors}{van der Heijden et al.}

\begin{abstract}
Geosocial reachability queries (\query{RangeReach}) determine whether a given vertex in a geosocial network can reach any spatial vertex within a query region. The state-of-the-art 3DReach method answers such queries by encoding graph reachability through interval labelling and indexing spatial vertices in a 3D R-tree. We present 2DReach, a simpler approach that avoids interval labelling entirely. Like 3DReach, 2DReach collapses strongly connected components (SCCs) into a DAG, but instead of computing interval labels, it directly stores a 2D R-tree per component over all reachable spatial vertices. A query then reduces to a single 2D R-tree lookup. We further propose compressed variants that reduce storage by excluding spatial sinks and sharing R-trees between components with identical reachable sets. Experiments on four real-world datasets show that 2DReach achieves faster index construction than 3DReach, with the compressed variant yielding the smallest index size among all methods. 2DReach delivers competitive or superior query performance with more stable response times across varying query parameters.
\end{abstract}


\maketitle

\section{Introduction}
Location-Based Social Networks (LBSNs) such as Yelp and Foursquare combine social interactions between users with spatial activities at physical venues. In these networks, users follow or befriend other users and check in at, review, or rate venues. The resulting graphs are \emph{geosocial}: vertices represent both users and venues, edges capture social and spatial relationships, and a subset of vertices carry geographic coordinates. A fundamental query on such networks is the \emph{Geosocial reachability query} (\query{RangeReach})~\cite{sun_georeach_2016}: given a query vertex $u$ and a spatial region $R$, determine whether $u$ can reach any spatial vertex located within $R$. \query{RangeReach} finds application in points-of-interest recommendation, geo-advertising, and epidemiological contact tracing~\cite{bouros_fast_2025}.

Existing methods for \query{RangeReach} face a tradeoff between the social (graph reachability) and spatial (range query) components of the problem. GeoReach~\cite{sun_georeach_2016} augments each vertex with spatial reachability summaries and prunes a graph traversal at query time, but its performance degrades when much of the graph must be explored. The state-of-the-art 3DReach~\cite{bouros_fast_2025} eliminates graph traversal by first collapsing strongly connected components (SCCs) into super-vertices, then computing interval labels that encode reachability in the resulting DAG. Spatial vertices are indexed in a 3D R-tree with the third dimension representing the post-order number used for interval-based reachability testing. While effective, this approach has two drawbacks: the interval labelling construction is costly, and the 3D R-tree requires larger bounding boxes (six floats per node versus four) and more complex range queries.

We propose 2DReach, which simplifies the index structure by avoiding interval labelling entirely. Like 3DReach, 2DReach collapses SCCs into super-vertices to form a DAG. However, instead of computing interval labels and lifting spatial data to 3D, 2DReach directly precomputes and stores, for each component, a 2D R-tree over all spatial vertices reachable from that component. A \query{RangeReach} query then reduces to a single 2D R-tree lookup in the query vertex's component. This design trades some additional R-tree storage for substantially simpler index construction and query execution.

We further introduce compressed variants that reduce storage by excluding spatial sinks from the SCC decomposition and sharing R-trees between components with identical reachable sets. Experiments on four real-world LBSN datasets show that 2DReach achieves faster index construction than 3DReach, with the most compressed variant yielding the smallest index size among all methods. 2DReach delivers competitive or superior query performance with substantially more stable response times across varying query parameters.

\section{Preliminaries}
This section introduces the necessary notations to understand geosocial reachability queries.

\paragraph{Geosocial Graph} A geosocial graph $G = (V, E, \delta)$ is a directed graph extended with spatial information. $V$ is the set of vertices and $E\subseteq V\times V$ is the set of directed edges. The function $\delta : V\rightarrow \mathbb{R}^{2}\cup\{\perp\}$ maps vertices to a 2-dimensional Euclidean point\footnote{Without loss of generality, this paper assumes 2D Euclidean space; for any other metric spaces, the solution can be easily extended.} or $\perp$ if the vertex has no spatial information. Lastly, $P = \{v\in V \ |\ \delta(v) \neq \perp\}$ denotes the set of spatial vertices. 

\paragraph{Queries} Given a directed graph $G = (V, E)$ and two vertices $u, v \in V$, a \textit{graph reachability query} ($u\overset{?}{\reach} v$) determines whether $u$ can reach $v$ iff there exists a path from $u$ to $v$. We denote this as $u\reach v$ or $u\not\reach v$, respectively. 

\begin{definition}[\textbf{RangeReach}]\label{def:range-reach}
    Given a geosocial graph $G = (V, E, \delta)$, a query vertex $u$, and an axis-aligned rectangle $R\subset \mathbb{R}^{2}$ (spatial region)\footnote{Without loss of generality, this paper assumes an axis-aligned rectangle for querying. However, the proposed method can be easily extended to handle other types of geometric objects for querying, e.g., polygons.}, the \textit{Geosocial reachability query} (\query{RangeReach}) determines whether $u$ can geosocially reach region $R$ iff there exists a path from $u$ to a vertex $v$ such that $\delta(v)$ is located in $R$. Formally:
    \[\operatorname{RangeReach}(G, u, R)=\begin{cases}
        \mathit{TRUE} & \text{if}\ \exists v\in V: u\reach v \land \delta(v) \in R\\
        \mathit{FALSE} & \text{otherwise}
    \end{cases}\]
\end{definition}

\paragraph{Strongly Connected Component Decomposition} Let $\mathcal{D} = (V_{\mathcal{D}}, E_{\mathcal{D}})$ be the Strongly Connected Component (SCC) decomposition of graph $G$. Each SCC of $G$ is represented by a single super-vertex in $V_{\mathcal{D}}$ and $E_{\mathcal{D}}$ contains a directed edge between two super-vertices if $G$ contains a directed edge between a pair of vertices of those SCCs. By construction, $\mathcal{D}$ is a directed acyclic graph (DAG).

\section{Related Work}
Sun and Sarwat~\cite{sun_georeach_2016} introduced the \query{RangeReach} query and proposed GeoReach, the first dedicated method for geosocial reachability. GeoReach augments every vertex in the graph with precomputed spatial reachability information in the form of a SPA-Graph, classifying vertices into three types of increasing precision: a single reachability bit, a reachability minimum bounding rectangle, or a set of hierarchical grid cells. Given a query, GeoReach traverses the SPA-Graph and uses these annotations to prune paths that cannot satisfy the spatial predicate. Although effective, GeoReach still requires graph traversal at query time; its performance degrades for queries with a negative answer and for graphs with many strongly connected components, as a large portion of the SPA-Graph may need to be explored.

Bouros et al.~\cite{bouros_fast_2025} proposed the state-of-the-art 3DReach method, which eliminates graph traversal entirely. 3DReach first collapses SCCs into super-vertices, then computes an interval-based labeling~\cite{agrawal_efficient_1989} that assigns each component a set of post-order intervals encoding its reachable descendants. Each spatial vertex is mapped to a three-dimensional point $(x, y, \mathit{post})$ and indexed by a 3D R-tree. A \query{RangeReach} query is answered by issuing one 3D range query per interval label of the query vertex. A variant, 3DReach-Rev, uses reversed labelling to represent spatial vertices as vertical line segments, reducing every query to a single 3D range query at the cost of indexing more complex geometric objects. However, interval labelling construction is costly, and 3D R-trees require larger bounding boxes and more complex queries than their 2D counterparts.

\begin{figure}[t]
    \centering
    \begin{subfigure}{\linewidth}
        \centering
        \begin{tikzpicture}[
            vertex/.style={circle, draw, thick, minimum size=12pt, inner sep=0pt, font=\scriptsize},
            spatial/.style={vertex, fill=gray!50},
            arr/.style={-{Stealth[length=4pt, width=3pt]}, thick},
            lbl/.style={font=\scriptsize},
            sccnode/.style={draw, thick, rounded corners=3pt, minimum size=14pt, inner sep=2pt, font=\scriptsize},
            sccbox/.style={draw, rounded corners=8pt, dashed, thick, inner sep=6pt, font=\scriptsize},
            x={(-0.4cm,-0.3cm)}, y={(0.9cm,-0.2cm)}, z={(0cm,0.8cm)}
        ]
            \fill[gray!15] (0,0,0) -- (5,0,0) -- (5,4,0) -- (0,4,0) -- cycle;
            \draw[thick] (0,0,0) -- (5,0,0) -- (5,4,0) -- (0,4,0) -- cycle;
            \fill[cyan!40, opacity=0.7] (2.0,2,0) -- (4.5,2,0) -- (4.5,4,0) -- (2.0,4,0) -- cycle;
            \draw[thick, cyan!70!black] (2.0,2,0) -- (4.5,2,0) -- (4.5,4,0) -- (2.0,4,0) -- cycle;
            \node[font=\small\bfseries] at (3.25,3,0) {$R$};
            
            \node[spatial] (f) at (1.3,-0.5,-1.2) {$f$};
            \node[spatial] (g) at (2,1.0,-0.2) {$g$};
            \node[spatial] (h) at (3.5,2.5,0.1) {$h$};
            \node[spatial] (i) at (4,3.8,0.2) {$i$};
            
            \node[vertex] (a) at (0.5,0.0,3.5) {$a$};
            \node[vertex] (b) at (1.5,-0.5,3) {$b$};
            \node[vertex] (c) at (1.5,1.5,3.0) {$c$};
            \node[vertex] (d) at (3.2,3.5,3.1) {$d$};
            \node[vertex] (e) at (3.2,4.5,2.6) {$e$};
            
            \node[
                sccbox,
                blue!60,
                fit=(a)(b)(c),
                label={[text=blue!70!black,font=\scriptsize]above:$C_1$}
                ] (C1) {};
            \node[
                sccbox,
                orange!80,
                fit=(d)(e),
                label={[text=orange!80!black,font=\scriptsize]above:$C_2$}
                ] (C2) {};
            
            \draw[arr] (a) -- (b);
            \draw[arr] (b) -- (c);
            \draw[arr] (c) -- (a);
            \draw[arr, bend left=20] (d) to (e);
            \draw[arr, bend left=20] (e) to (d);
            \draw[arr] (c) -- (d);
            
            \draw[arr, densely dashed, gray] (a) -- (f);
            \draw[arr, densely dashed, gray] (c) -- (g);
            \draw[arr, densely dashed, gray] (d) -- (h);
            \draw[arr, densely dashed, gray] (e) -- (i);
            
            \node[lbl, anchor=south, blue!70!black] at ($(C1.north)+(2,0,0.7)$) {$[0,5]$};
            \node[lbl, anchor=south, orange!80!black] at ($(C2.north)+(1.5,0,0.5)$) {$[1,3]$};
        \end{tikzpicture}
        \caption{}
    \end{subfigure}

    \vspace{0.8cm}

    \begin{subfigure}{\linewidth}
        \centering
        \begin{tikzpicture}[
            arr/.style={-{Stealth[length=4pt, width=3pt]}, thick},
            sccnode/.style={draw, thick, rounded corners=3pt, minimum size=14pt, inner sep=2pt, font=\scriptsize}
        ]
            \node[sccnode, draw=blue!70!black, thick] (sc1) at (0.8, 0) {$C_1$};
            \node[sccnode, draw=orange!80!black, thick] (sc2) at (3.5, 0) {$C_2$};
            \draw[arr] (sc1) -- (sc2);
            \node[font=\scriptsize, anchor=north, blue!70!black] at (0.8, -0.4) {2D R-tree$(\{f,g,h,i\})$};
            \node[font=\scriptsize, anchor=north, orange!80!black] at (3.5, -0.4) {2D R-tree$(\{h,i\})$};
        \end{tikzpicture}
        \caption{}
    \end{subfigure}

    \caption{Running example. (a)~Geosocial graph~$G$ with SCCs $C_1{=}\{a,b,c\}$ and $C_2{=}\{d,e\}$; spatial vertices (gray) on the 2D plane; query region~$R$ (cyan). Interval labels (shown per component) are used by 3DReach for reachability encoding. (b)~2DReach index: SCC DAG with one 2D R-tree per component storing reachable spatial vertices.}
    \label{fig:running-example}
\end{figure}
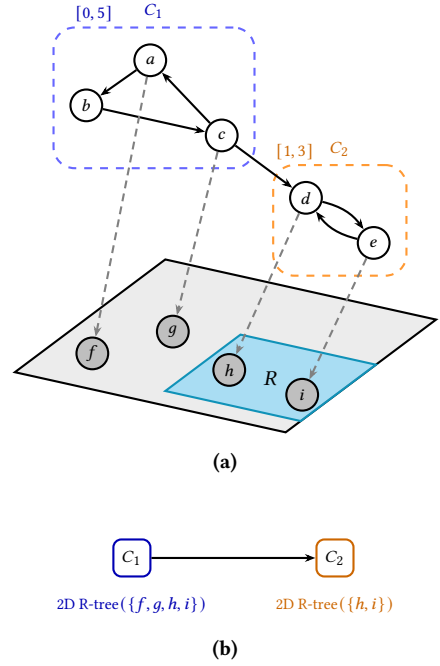

2DReach shares with 3DReach the property of avoiding graph traversal at query time, but takes a simpler approach. Both methods collapse SCCs into super-vertices to form a DAG. However, 3DReach then computes interval labels (\cref{fig:running-example}(a)) and indexes spatial vertices in a 3D R-tree, requiring one 3D range query per interval at query time. In contrast, 2DReach avoids interval labelling entirely: it directly precomputes and stores a 2D R-tree per component over the reachable spatial vertices (\cref{fig:running-example}(b)). This reduces each query to a single 2D R-tree lookup, with smaller bounding boxes and simpler range predicates.

\section{2DReach}\label{sec:2dreach}
In this section, we introduce 2DReach. First, we explain the baseline approach, and then we describe memory compression techniques.

Like 3DReach, 2DReach computes an SCC decomposition $\mathcal{D}$ of the input graph $G$ to handle cycles. However, instead of computing interval labels and indexing spatial vertices in a 3D R-tree, 2DReach directly stores a 2D R-tree\footnote{For higher-dimensional metric spaces, other spatial indexes can be used, e.g.\ higher-dimensional R-trees.} for each node $n$ in $\mathcal{D}$, containing all spatial vertices reachable from $n$.

To construct the index, 2DReach first computes the reachable spatial vertices for each node in $\mathcal{D}$. Subsequently, it traverses $\mathcal{D}$ in reverse topological order. Upon reaching a node, its reachable spatial vertices are obtained by merging those of its immediate children with its own. Lastly, for each node, an R-tree is built on the reachable spatial vertices. \cref{alg:2dreach-construction} outlines this process.

To query this structure, 2DReach takes a query vertex $u$ and region $R$ as input. Initially, it identifies the strongly connected component $c$ in $\mathcal{D}$ that contains $u$. Afterwards, 2DReach queries the R-tree that belongs to $c$ to determine whether it contains any points in the query region $R$. Finally, it returns the result of this query. The standard 2DReach in ~\cref{sec:experiments} employs pointers, for each node, to the R-tree of $c$ to improve query time.

As an illustration, consider the geosocial graph in \cref{fig:running-example}(a) with SCCs $C_1 = \{a,b,c\}$ and $C_2 = \{d,e\}$. Processing $\mathcal{D}$ in reverse topological order, 2DReach first builds an R-tree for $C_2$ over $\{h,i\}$, then for $C_1$ over $\{f,g\} \cup \{h,i\}$ (\cref{fig:running-example}(b)). To answer $\query{RangeReach}(G, a, R)$, 2DReach maps $a$ to $C_1$ and queries its R-tree with region $R$, returning TRUE since $h \in R$.

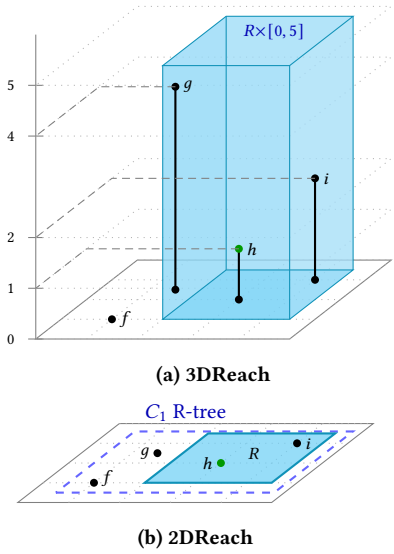
\begin{figure}[t]
\centering
\begin{tikzpicture}[
    scale=0.75,
    x={(0.9cm,0cm)},
    y={(0.45cm,0.35cm)},
    z={(0cm,0.9cm)},
    point/.style={circle, fill=black, minimum size=3pt, inner sep=0pt}
]
\draw[thin, gray] (0,0,0) -- (5,0,0) -- (5,4,0) -- (0,4,0) -- cycle;
\foreach \i in {1,2,3,4} {
    \draw[dotted, gray!60] (\i,0,0) -- (\i,4,0);
    \draw[dotted, gray!60] (0,\i,0) -- (5,\i,0);
}
\draw[thin, gray] (0,0,0) -- (0,0,6);
\foreach \z in {0,1,2,4,5} {
    \draw[thin, gray] (-0.15,0,\z) -- (0.15,0,\z);
    \node[font=\scriptsize, anchor=east] at (-0.25,0,\z) {\z};
    \draw[dotted, gray!60] (0,0,\z) -- (5,0,\z) -- (5,4,\z) -- (0,4,\z) -- (0,0,\z);
}

\fill[cyan!50, opacity=0.6] (2,1,0) -- (4.5,1,0) -- (4.5,3.5,0) -- (2,3.5,0) -- cycle;
\fill[cyan!40, opacity=0.6] (2,1,5) -- (4.5,1,5) -- (4.5,3.5,5) -- (2,3.5,5) -- cycle;
\fill[cyan!45, opacity=0.6] (2,1,0) -- (4.5,1,0) -- (4.5,1,5) -- (2,1,5) -- cycle;
\fill[cyan!35, opacity=0.6] (4.5,1,0) -- (4.5,3.5,0) -- (4.5,3.5,5) -- (4.5,1,5) -- cycle;
\fill[cyan!30, opacity=0.4] (2,3.5,0) -- (4.5,3.5,0) -- (4.5,3.5,5) -- (2,3.5,5) -- cycle;
\fill[cyan!35, opacity=0.5] (2,1,0) -- (2,3.5,0) -- (2,3.5,5) -- (2,1,5) -- cycle;

\draw[cyan!70!black] (2,1,0) -- (4.5,1,0) -- (4.5,3.5,0) -- (2,3.5,0) -- cycle;
\draw[cyan!70!black] (2,1,5) -- (4.5,1,5) -- (4.5,3.5,5) -- (2,3.5,5) -- cycle;
\draw[cyan!70!black] (2,1,0) -- (2,1,5);
\draw[cyan!70!black] (4.5,1,0) -- (4.5,1,5);
\draw[cyan!70!black] (4.5,3.5,0) -- (4.5,3.5,5);
\draw[cyan!70!black] (2,3.5,0) -- (2,3.5,5);

\node[point] at (1,1,0) {};
\node[font=\scriptsize, anchor=west] at (1,1,0) {$f$};

\draw[thick] (1.5,2.5,0) -- (1.5,2.5,4);
\node[point] at (1.5,2.5,0) {};
\node[point] at (1.5,2.5,4) {};
\node[font=\scriptsize, anchor=west] at (1.5,2.5,4) {$g$};
\draw[thin, densely dashed, gray] (1.5,2.5,4) -- (0,2.5,4) -- (0,0,4);

\draw[thick] (3,2,0) -- (3,2,1);
\node[point] at (3,2,0) {};
\node[point, fill=green!60!black] at (3,2,1) {};
\node[font=\scriptsize, anchor=west] at (3,2,1) {$h$};
\draw[thin, densely dashed, gray] (3,2,1) -- (0,2,1) -- (0,0,1);

\draw[thick] (4,3,0) -- (4,3,2);
\node[point] at (4,3,0) {};
\node[point] at (4,3,2) {};
\node[font=\scriptsize, anchor=west] at (4,3,2) {$i$};
\draw[thin, densely dashed, gray] (4,3,2) -- (0,3,2) -- (0,0,2);

\node[font=\scriptsize, blue!70!black] at (4.1,1.2,5.6) {$R{\times}[0,5]$};
\node[font=\small\bfseries] at (2.5,2,-1.5) {(a) 3DReach};
\end{tikzpicture}
\hspace{0.5cm}
\begin{tikzpicture}[
    scale=0.75,
    x={(0.9cm,0cm)},
    y={(0.45cm,0.35cm)},
    z={(0cm,0.9cm)},
    point/.style={circle, fill=black, minimum size=3pt, inner sep=0pt}
]
\draw[thin, gray] (0,0,0) -- (5,0,0) -- (5,4,0) -- (0,4,0) -- cycle;
\foreach \i in {1,2,3,4} {
    \draw[dotted, gray!60] (\i,0,0) -- (\i,4,0);
    \draw[dotted, gray!60] (0,\i,0) -- (5,\i,0);
}

\draw[blue!60, thick, dashed] (0.5,0.5,0) -- (4.8,0.5,0) -- (4.8,3.6,0) -- (0.5,3.6,0) -- cycle;
\node[font=\small, blue!70!black, anchor=south west] at (0.5,3.7,0) {$C_1$ R-tree};

\fill[cyan!50, opacity=0.7] (2,1,0) -- (4.5,1,0) -- (4.5,3.5,0) -- (2,3.5,0) -- cycle;
\draw[cyan!70!black, thick] (2,1,0) -- (4.5,1,0) -- (4.5,3.5,0) -- (2,3.5,0) -- cycle;
\node[font=\scriptsize] at (3.4,2.5,0) {$R$};

\node[point] at (1,1,0) {};
\node[font=\scriptsize, anchor=west] at (1,1,0.1) {$f$};

\node[point] at (1.5,2.5,0) {};
\node[font=\scriptsize, anchor=east] at (1.5,2.5,0) {$g$};

\node[point, fill=green!60!black] at (3,2,0) {};
\node[font=\scriptsize, anchor=east] at (3,2,0) {$h$};

\node[point] at (4,3,0) {};
\node[font=\scriptsize, anchor=west] at (4,3,0) {$i$};

\node[font=\small\bfseries] at (2.5,2,-1.5) {(b) 2DReach};
\end{tikzpicture}
\caption{Query execution for $\query{RangeReach}(G, a, R)$. (a)~3DReach extends region $R$ into a 3D cuboid using interval $[0,5]$. (b)~2DReach queries $C_1$'s 2D R-tree directly. Both find $h \in R$ (highlighted), but 2DReach avoids the third dimension.}
\label{fig:query-comparison}
\end{figure}

\begin{algorithm}[htbp]
    \caption{2DReach construction}
    \label{alg:2dreach-construction}
    \begin{algorithmic}[1]
        \Input Geosocial Graph $G = (V, E, \delta)$, SCC Decomposition $\mathcal{D}$
        \Output 2DReach index 
        \algseparator
        \State Let $S$ be a map from nodes of $\mathcal{D}$ to a subset of nodes of $V$.
        \State Let $R$ be a map from nodes of $\mathcal{D}$ to a 2D-R-tree.
        \ForAll{nodes $n\in \mathcal{D}$}
            \State Store in $S[n]$ all spatial nodes that are in $n$.
        \EndFor
        \State Let $\pi$ be the reverse topological order of nodes on $\mathcal{D}$.
        \For{$n\in \pi$}
            \State $\mathcal{C}\gets$ Children of $n$.
            \State $S[n]\gets S[n] \cup \bigcup\limits_{c\in\mathcal{C}}S[c]$.
            \State $R[n]\gets$ 2D-R-tree on $S[n]$.
        \EndFor
        \State \Return $R$
    \end{algorithmic}
\end{algorithm}

\subsection{Compressed 2DReach}\label{sec:compressed-2dreach}
Having discussed the baseline 2DReach, this section addresses strategies to reduce 2DReach's space usage. First, observe that LBSNs often contain spatial nodes without outgoing edges. Therefore, to reduce the space requirement, the SCC decomposition is constructed solely on the subgraph induced by the non-spatial vertices (the social subgraph). Consequently, this modification requires changing line 4 of \cref{alg:2dreach-construction} to include any spatial node that is a direct neighbour of node $n$. During query execution, 2DReach verifies whether each spatial vertex lies within the query region; the corresponding procedure is outlined in \cref{alg:2dreach-querying}.

This extension can be further generalised to include any strongly connected component in the social subgraph that has no reachable spatial vertices. This method also assumes that spatial vertices have no outgoing edges, although this is not a strict constraint in the data model. Therefore, in other data models, the compression procedure should be confined to spatial vertices without outgoing edges. As a consequence, line 1 of \cref{alg:2dreach-querying} should also verify whether $q$ has outgoing edges.

Additionally, when traversing in topological order, a parent may have the same reachable set of spatial vertices as one of its children. To further compress the 2DReach index, any such parent will share its R-tree with that child. Finally, 2DReach binds R-tree pointers to graph nodes to improve query performance; one method in \cref{sec:experiments} will explore using pointers at the SCC level, as outlined in \cref{alg:2dreach-querying}, to compress the index.

\begin{algorithm}[htbp]
    \caption{2DReach-Compressed querying}
    \label{alg:2dreach-querying}
    \begin{algorithmic}[1]
        \Require 2DReach Index $\mathcal{R}$, Map from node to SCC decomposition node $\mathcal{M}$
        \Input Geosocial Graph $G = (V, E, \delta)$, query vertex $q$, query region $R$
        \Output Boolean answer to \query{RangeReach}
        \algseparator
        \If{$q$ is a spatial vertex}
            \State \Return $\delta(q)\in R$
        \Else
            \State $c\gets \mathcal{M}[q]$
            \State \Return $\mathcal{R}[c][R]$
        \EndIf
    \end{algorithmic}
\end{algorithm}

\subsection{Complexity Analysis}\label{sec:complexity-analysis}
This section discusses the theoretical worst-case performance of the introduced 2DReach method. However, the introduction of 2DReach focuses on practical performance; therefore, see \cref{sec:experiments} for an experimental evaluation of 2DReach's practical performance. 

Let $G = (V, E, \delta)$ be a geosocial graph with set $P$ the spatial vertices and let $\mathcal{D} = (V_{\mathcal{D}}, E_{\mathcal{D}})$ be the SCC decomposition of $G$. Then, we define the cardinalities as follows $d = |V_{\mathcal{D}}|$, $e = |E_{\mathcal{D}}|$, $n = |V|$, $m = |E|$, and $p = |P|$.

\begin{theorem}
    Storing the 2DReach index uses $\bigO(d \cdot p)$ space.
\end{theorem}
\begin{proof}
    An $n$-element 2D R-tree uses $\bigO(n)$ space. 2DReach stores one 2D R-tree, with at most $p$ nodes, per SCC node. Lastly, at most one pointer per SCC node.
\end{proof}
\begin{remark}
    The worst-case space complexity remains equal for the compressed 2DReach version.
\end{remark}

\begin{theorem}
 The construction of the 2DReach index takes $\mathcal{O}(n + d\cdot p\cdot (d + \log(p)))$ time.
\end{theorem}
\begin{proof}
    Storing the spatial nodes in $S$ requires verifying for each node whether it is spatial, $\bigO(n)$. Moreover, the topological order requires linear time, $\bigO(d + e)$, over the SCC decomposition. Note that obviously $e\in\bigO(d^2)$. Furthermore, for each node of $\mathcal{D}$, computing $\mathcal{S}$ might require $\bigO(p)$ set insertions for $\bigO(d)$ children. Lastly, for each node of $\mathcal{D}$, a 2D R-tree is built, which requires $\bigO(p\cdot \log(p))$ time, as each tree has at most $p$ elements.
\end{proof}
\begin{remark}
    In the worst-case time complexity of the compressed 2DReach version, $d$ can be replaced by $d - p$, provided that all spatial vertices are sinks. However, without this assumption, the worst-case time complexity does not differ.
\end{remark}

\begin{theorem}
    The querying time of the 2DReach index is $\bigO(p)$.
\end{theorem}
\begin{proof}
    Equivalent to R-tree querying time.
\end{proof}
\begin{remark}
    The average querying time for an $n$-element R-tree is $\bigO(\log_{M}(n))$ with $M$ the maximum entries per page. Subsequently, the average querying time for 2DReach is equivalent. 
\end{remark}

\begin{figure*}[t!]
    \centering
    \includegraphics[width=\linewidth]{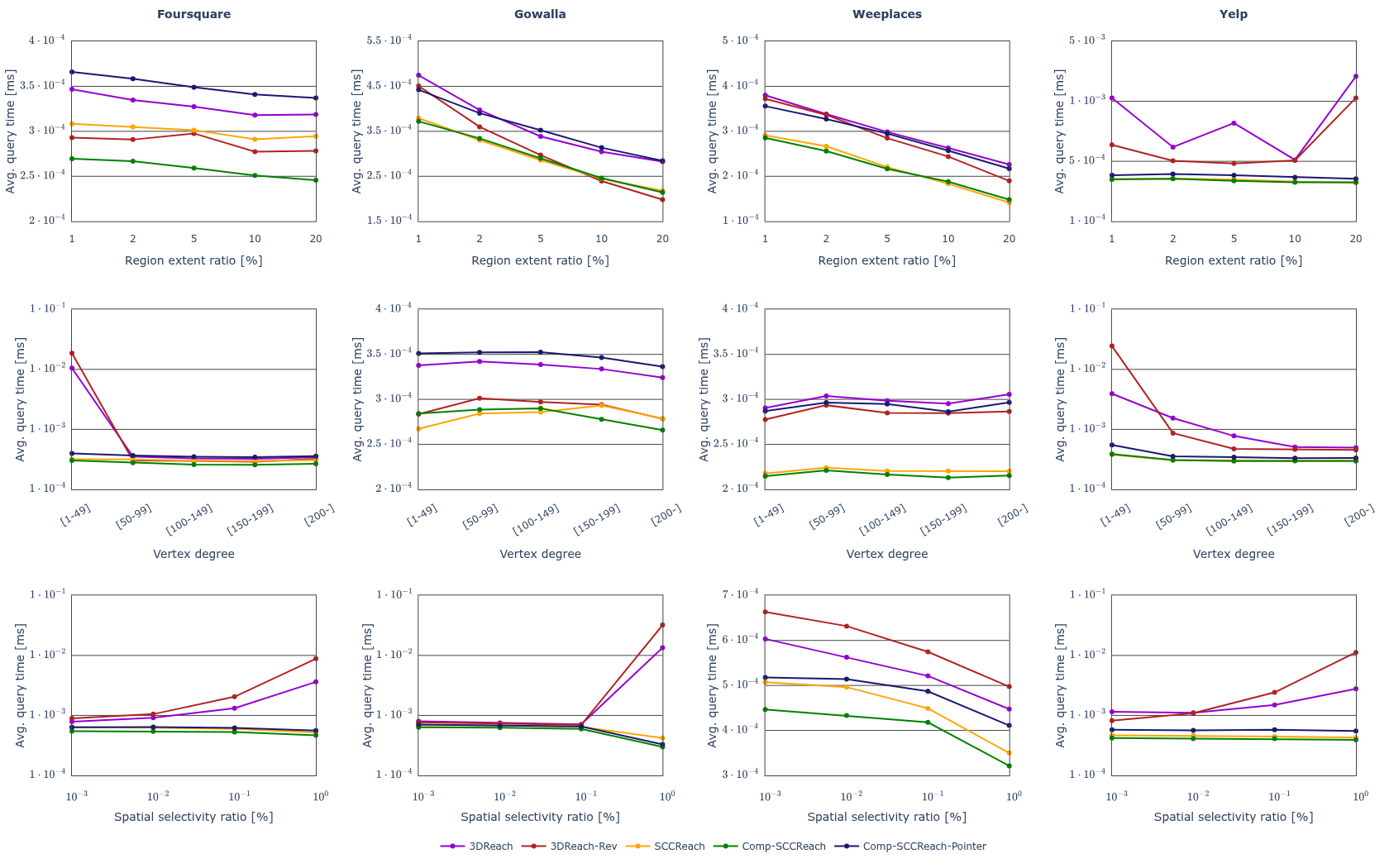}
    \caption{Experimental results displayed in line charts. The columns represent the four datasets, the rows represent the three experimental parameters. The Yelp (all), Foursquare (vertex degree \& spatial selectivity), and Gowalla (spatial selectivity) are displayed on a logarithmic axis; all other plots are illustrated using a linear axis.}
    \label{fig:results_mean}
\end{figure*}

\begin{table*}[t!]
    \centering
    \begin{tabular}{lcc}
        \toprule
        Parameter & (Experimental) Values & Default Value \\ \midrule
        Region Extent Ratio & $1\%, 2\%, 5\%, 10\%, 20\%$ & $5\%$ \\
        Vertex Degree & $[1-49], [50-99], [100-149], [150-199], [200-]$ & $[100-149]$ \\
        Spatial Selectivity Ratio & $0.001\%, 0.01\%, 0.1\%, 1\%$ & N/A \\
        \bottomrule
    \end{tabular}
    \caption{The three parameters for which to evaluate experiments. For each parameter, the (experimental) values are executed. Moreover, in other experiments, the default values are taken.}
    \label{tab:experimental-parameters}
\end{table*}

\begin{table*}[t!]
    \centering
    \begin{tabular}{lccccccc}
        \toprule
        dataset & Users & Venues & Nodes & Edges & SCCs & User SCCs & distinct R-trees \\ \midrule
        Foursquare & 2,119,987 (65.2\%) & 1,132,617 (34.8\%) & 3,252,604 & 19,685,786 & 1,400,154 & 267,537 (19.1\%) & 267,535 \\
        Gowalla & 407,533 (13.0\%) & 2,723,102 (87.0\%) & 3,130,635 & 23,778,362 & 2,723,103 & 1 (0.0\%) & 1 \\
        Weeplaces & 16,022 (1.6\%) & 971,309 (98.4\%) & 987,331 & 2,758,946 & 971,311 & 2 (0.0\%) & 1 \\
        Yelp & 1,987,693 (93.0\%) & 150,310 (7.0\%) & 2,138,003 & 21,357,271 & 1,238,535 & 1,088,225 (87.9\%) & 1,088,225 \\
        \bottomrule
    \end{tabular}
    \caption{Graph and SCC statistics per dataset. The users and venues are non-spatial and spatial nodes, respectively. The percentage of users/venues are compared to nodes. The user SCCs are SCCs without spatial information and their percentage is compared to SCCs. The distinct R-trees are those that have distinct spatial reachability compared to their children.}
    \label{tab:dataset-statistics}
\end{table*}

\begin{table*}[t!]
    \centering
    \begin{tabular}{lccccc}
        \toprule
        dataset & 3DReach & 3DReach-Rev & 2DReach & 2DReach-Comp & 2DReach-Pointer \\ \midrule
        Foursquare & 4.841 & 5.928 & 3.529 & 3.257 & \textbf{3.241} \\
        Gowalla & 7.909 & 8.764 & 5.400 & 4.656 & \textbf{4.639} \\
        Weeplaces & 1.225 & 1.745 & 0.803 & 0.681 & \textbf{0.676} \\
        Yelp & 4.372 & 5.427 & 2.663 & 2.669 & \textbf{2.661} \\
        \bottomrule
    \end{tabular}
    \caption{Index construction time [secs]: To build the SCC decomposition for each method. Also, to build the (reverse) interval labelling and 3D R-tree for 3DReach and 3DReach-Rev, the R-trees and the pointer (and bit vectors) for the 2DReach methods.}
    \label{tab:index-time-methods}
\end{table*}

\begin{table*}[t!]
    \centering
    \begin{tabular}{lccccc}
        \toprule
        dataset & 3DReach & 3DReach-Rev & 2DReach & 2DReach-Comp & 2DReach-Pointer \\ \midrule
        Foursquare & 35.1 (20.5 / 14.6) & 78.5 & 49.3 (23.3 / 26.0) & 40.3 (14.2 / 26.0) & \textbf{16.6} (14.2 / 2.4) \\
        Gowalla & 71.2 (49.4 / 21.8) & 164.1 & 73.0 (47.9 / 25.0) & 51.2 (26.1 / 25.0) & \textbf{26.7} (26.1 / 0.5) \\
        Weeplaces & 25.4 (17.6 / 7.8) & 58.5 & 25.0 (17.1 / 7.9) & 17.2 (9.3 / 7.9) & \textbf{9.5} (9.3 / 0.2) \\
        Yelp & 31.4 (2.7 / 28.7) & 79.8* & 39.0 (21.9 / 17.1) & 37.8 (20.7 / 17.1) & \textbf{29.7} (20.7 / 8.9) \\
        \bottomrule
    \end{tabular}
    \caption{Index size [MBs]. Values in parentheses show (R-tree storage / labelling storage) for 3DReach and (R-tree storage / pointer storage) for 2DReach variants. *Larger than reported in~\cite{bouros_fast_2025} (44.8 MB), likely due to differences in Boost library versions.}
    \label{tab:index-size-methods}
\end{table*}

\section{Experiments}\label{sec:experiments}
This section discusses the experimental evaluation of the 2DReach. First, the experimental setup is discussed. Afterwards, the experimental results are evaluated. All experiments are performed on a laptop with 64GB of RAM, an Intel(R) Core(TM) Ultra 9 275HX CPU clocked at 3.80 GHz, and Ubuntu 24.04 LTS. Moreover, the code is written in C++17 and compiled with gcc (v13.3.0) with the \textsc{-O3} flag enabled and available on GitHub\footnote{\codeLink}.

\subsection{Experimental Setup}
For the experiments, we utilise four real-world datasets: Yelp\footnote{www.yelp.com/dataset}, Foursquare~\cite{levandoski_lars_2012,sarwat_lars_2014}\footnote{https://academictorrents.com/details/b24c73949308b3f6bdd8fea1a485534392eef338}, Gowalla~\cite{liu_exploiting_2014}\footnote{The Gowalla \& Weeplaces datasets are downloaded from the download url from 3DReach~\cite{bouros_fast_2025}: https://seafile.rlp.net/d/6fa3703e57234b6d9831/.}, and Weeplaces\footnotemark[\value{footnote}]. The datasets are Location-Based Social Networks (LBSNs), which consist of a social graph of users (nodes without spatial information) and venues (nodes with spatial information). Moreover, for these types of graphs, the venues have no outgoing edges as they can only be visited, checked into, or reviewed. \cref{tab:dataset-statistics} presents statistics regarding the datasets. These datasets are also used for the 3DReach evaluation~\cite{bouros_fast_2025}.

During the experiments, we evaluate the following methods:
\begin{itemize}
    \item \textsc{3DReach}: The state-of-the-art 3DReach index from~\cite{bouros_fast_2025}. We use the source code provided by the authors\footnote{https://github.com/pbour/rangereach}.
    \item \textsc{3DReach-Rev}: A variant of 3DReach that uses reversed interval labeling to represent spatial vertices as vertical line segments, reducing each query to a single 3D range query~\cite{bouros_fast_2025}. We use the source code provided by the authors\footnotemark[\value{footnote}]. 
    \item \textsc{2DReach}: The standard 2DReach described in \cref{sec:2dreach} without compression mechanisms. 
    \item \textsc{2DReach-Comp}: The 2DReach described in \cref{sec:compressed-2dreach} that employs the various compression mechanisms.
    \item \textsc{2DReach-Pointer}: The 2DReach-Comp with the addition that the pointers from the SCC decomposition nodes are only stored for non-spatial components instead of nodes, and utilise a bit vector to retrieve the correct R-tree.
\end{itemize}

A series of experiments is conducted across distinct parameters to evaluate the performance of the methods. For each parameter, there are 1000 queries executed per value. Each such experiment is repeated ten times, and the median is reported as the final result. The different parameters for the experiments are the following:
\begin{description}
    \item[Region Extent Ratio] The region extent ratio is the size of the query region as a percentage of the entire spatial region. 
    \item[Vertex Degree] The vertex degree is the outgoing degree of the query vertex within a specified range.
    \item[Spatial Selectivity Ratio] The spatial selectivity ratio is the number of spatial nodes inside the query region as a percentage of the number of nodes in the geosocial graph.
\end{description}

For each parameter, \cref{tab:experimental-parameters} lists the values for which the experiments are executed. The default values are subsumed when executing experiments for the other parameters, except for the spatial selectivity ratio. The latter can take any value when conducting experiments with the other parameters. We follow the experimental methodology established by 3DReach~\cite{bouros_fast_2025} to enable direct comparison.

\subsection{Evaluation}
This section evaluates the different methods based on the described experiments. First, the initialisation process and the storage overhead are discussed. Then, the query performance between different methods is evaluated.

\subsubsection{Initialization \& Storage Overhead}
All methods require an initialisation phase to build the index for a geosocial graph and store it for query execution. Each method requires an SCC decomposition. Furthermore, the 3DReach and 3DReach-Rev methods also create an interval labelling scheme, which the former stores. Additionally, the 2DReach methods require building and storing pointers from nodes to the R-trees. Moreover, the 2DReach-Pointer also stores extra bit vectors. Lastly, all methods need to build and store their R-tree(s).

\cref{tab:index-time-methods} presents the durations, in seconds, for the initialisation phase of each method. The results show that all 2DReach methods outperform the construction time of 3DReach and 3DReach-Rev. The 3DReach and 3DReach-Rev construction times are primarily dominated by their labelling, which accounts for the noticeable differences compared to other methods. 

\cref{tab:index-size-methods} provides an overview of the index sizes. It indicates that the storage overhead of the standard 2DReach is already only within at most 40\% additional overhead. However, when social complexity increases significantly (Yelp and Foursquare), this method becomes less competitive. This is due to the increase in social SCCs, which require more complex 2D R-trees. 2DReach-Pointer improves substantially across all cases, as storage overhead increases with more social components in standard 2DReach and non-spatial components are significantly less than the node count.

As the number of spatial nodes increases, 2DReach-Comp improves over 2DReach (Gowalla and Weeplaces). Due to a decrease in R-trees necessary for spatial sinks. \cref{tab:dataset-statistics} further reveals that reusing children's R-trees has little effect, reducing the number by at most two. However, in low-SCC cases such as Weeplaces, this has a more consequential impact. Finally, due to compression and pointer reduction, 2DReach-Pointer continues to outperform all other methods in terms of storage overhead.

Notably, 2DReach-Pointer achieves smaller index sizes than both 3DReach variants despite storing multiple R-trees. This is because 3DReach must store interval labels for all super vertices, which can amount to millions of intervals in large graphs. In contrast, 2DReach leverages the social SCC membership directly by forgoing 2D R-tree structures for spatial sinks and avoiding interval labelling entirely. Furthermore, 2DReach-Pointer shares R-trees between parent and children when possible, and its 2D R-tree nodes require smaller bounding boxes (4 floats) compared to 3D R-tree nodes (6 floats). The combination of R-tree sharing, exclusion of spatial sinks from the decomposition, and elimination of interval labels explains why the theoretical worst-case space of 2DReach does not materialise in practice.

\subsubsection{Query Time}
Finally, we analyse the query performance of the different methods. The results of these experiments are illustrated in \cref{fig:results_mean}. All Yelp experiments, Foursquare vertex degree and spatial selectivity ratio experiments, and Gowalla spatial selectivity ratio experiment are depicted with a logarithmic scale. All other experiments are shown on a linear scale to accommodate differences in query times. 

Firstly, across nearly all region extent ratio values, except for Yelp, both 3DReach(-Rev) achieve their best performance. This also applies to the Weeplaces experiments and the Gowalla vertex-degree experiments. In these cases, the 2DReach methods still perform up to 30\% better. This difference is explained by performing simpler 2D queries instead of 3D queries.

Other experiments show a more exponential behaviour from 3DReach(-Rev). In the Foursquare vertex degree and Gowalla spatial selectivity experiments, both 3DReach methods perform on par with 2DReach. However, for specific parameter values, their query performance is orders of magnitude slower. In contrast, 2DReach mitigates this exponential growth and maintains relatively consistent query performance regardless of parameter values.

In all Yelp experiments and the Foursquare selectivity experiment, the 3DReach(-Rev) are generally an order of magnitude slower than the 2DReach variants. A probable reason is that, irrespective of the graph's structure, 2DReach queries only a 2D R-tree. The most significant advantage of 2DReach is its ability to stabilise and optimise query performance across all experiments, preventing substantial degradation in query performance. 

Finally, observe that 2DReach and 2DReach-Comp exhibit a similar query performance, although 2DReach-Pointer is consistently slower than both. This difference can reach up to a factor of 30\%, and 2DReach-Pointer loses the advantage in high-performance experiments compared to 3DReach(-Rev), unlike other 2DReach variants. However, when the 3DReach(-Rev) performance declines exponentially, 2DReach-Pointer remains comparable to 2DReach and 2DReach-Comp. Evidently, the loss in query performance is due to the need to retrieve the R-tree pointer via the SCC node and bit vectors.

\section{Conclusion}\label{sec:conclusion}
We introduced 2DReach, a simpler, faster, and smaller alternative to the state-of-the-art 3DReach~\cite{bouros_fast_2025} for geosocial reachability queries. By eliminating interval labelling, 2DReach replaces 3DReach's 3D R-tree with per-component 2D R-trees, simplifying both index construction and query execution. Compressed variants further reduce storage by excluding spatial sinks and sharing R-trees between components with identical reachable sets. Experiments on four real-world datasets show that 2DReach achieves faster index construction, with 2DReach-Pointer yielding the smallest index size among all methods while delivering stable query performance, up to an order of magnitude faster than 3DReach on datasets with high social complexity.

Future work could explore the applicability of 2DReach to other spatial graph domains and investigate persistent R-trees to improve R-tree sharing across components.

\section*{Artifact Availability}
The source code and experimental scripts are available at \codeLink.

\bibliographystyle{ACM-Reference-Format}
\bibliography{literature}

@inproceedings{agrawal_efficient_1989,
    title = {Efficient Management of Transitive Relationships in Large Data and Knowledge Bases},
    doi = {10.1145/67544.66950},
    booktitle = {Proceedings of the 1989 {ACM} {SIGMOD} International Conference on Management of Data},
    address = {Portland, Oregon, USA},
    publisher = {ACM Press},
    author = {Agrawal, Rakesh and Borgida, Alexander and Jagadish, H. V.},
    year = {1989},
    pages = {253--262},
}

@misc{bouros_fast_2025,
    title = {Fast {Geosocial} {Reachability} {Queries}},
    url = {https://openproceedings.org/2025/conf/edbt/paper-13.pdf},
    doi = {10.48786/EDBT.2025.03},
    language = {en},
    urldate = {2025-06-03},
    publisher = {OpenProceedings.org},
    author = {Bouros, Panagiotis and Chondrogiannis, Theodoros and Kowalski, Daniel},
    year = {2025},
}

@misc{sun_georeach_2016,
    title = {{GeoReach}: {An} {Efficient} {Approach} for {Evaluating} {Graph} {Reachability} {Queries} with {Spatial} {Range} {Predicates}},
    shorttitle = {{GeoReach}},
    url = {http://arxiv.org/abs/1603.05355},
    doi = {10.48550/arXiv.1603.05355},
    urldate = {2025-07-08},
    publisher = {arXiv},
    author = {Sun, Yuhan and Sarwat, Mohamed},
    month = mar,
    year = {2016},
    note = {arXiv:1603.05355 [cs]},
}

@inproceedings{levandoski_lars_2012,
    address = {USA},
    series = {{ICDE} '12},
    title = {{LARS}: {A} {Location}-{Aware} {Recommender} {System}},
    isbn = {978-0-7695-4747-3},
    shorttitle = {{LARS}},
    url = {https://doi.org/10.1109/ICDE.2012.54},
    doi = {10.1109/ICDE.2012.54},
    urldate = {2026-01-21},
    booktitle = {Proceedings of the 2012 {IEEE} 28th {International} {Conference} on {Data} {Engineering}},
    publisher = {IEEE Computer Society},
    author = {Levandoski, Justin J. and Sarwat, Mohamed and Eldawy, Ahmed and Mokbel, Mohamed F.},
    month = apr,
    year = {2012},
    pages = {450--461},
}

@article{sarwat_lars_2014,
    title = {{LARS}*: {An} {Efficient} and {Scalable} {Location}-{Aware} {Recommender} {System}},
    volume = {26},
    issn = {1558-2191},
    shorttitle = {{LARS}*},
    url = {https://ieeexplore.ieee.org/document/6427747},
    doi = {10.1109/TKDE.2013.29},
    number = {6},
    urldate = {2026-01-21},
    journal = {IEEE Transactions on Knowledge and Data Engineering},
    author = {Sarwat, Mohamed and Levandoski, Justin J. and Eldawy, Ahmed and Mokbel, Mohamed F.},
    month = jun,
    year = {2014},
    pages = {1384--1399},
}

@inproceedings{liu_exploiting_2014,
    address = {New York, NY, USA},
    series = {{CIKM} '14},
    title = {Exploiting {Geographical} {Neighborhood} {Characteristics} for {Location} {Recommendation}},
    isbn = {978-1-4503-2598-1},
    url = {https://dl.acm.org/doi/10.1145/2661829.2662002},
    doi = {10.1145/2661829.2662002},
    urldate = {2026-01-21},
    booktitle = {Proceedings of the 23rd {ACM} {International} {Conference} on {Conference} on {Information} and {Knowledge} {Management}},
    publisher = {Association for Computing Machinery},
    author = {Liu, Yong and Wei, Wei and Sun, Aixin and Miao, Chunyan},
    month = nov,
    year = {2014},
    pages = {739--748},
}

\end{document}